\documentclass[12pt]{amsart}
\usepackage[english]{babel}
\usepackage{amsmath,amsfonts,amssymb,amsthm}
\usepackage{graphicx}
\usepackage{color}
\newcommand{\IR}{{\mathbb{R}}}
\newcommand{\ID}{{\mathbb{D}}}
\newcommand{\ZZ}{{\mathbb{Z}}}

\def\le{\leqslant}
\def\ge{\geqslant}

\newtheorem{theorem}{Theorem}

\newtheorem{lemma}{Lemma}[section]
\newtheorem{prop}[lemma]{Proposition}

\theoremstyle{definition}

\let\llldots=\ldots
\def\ldots{\llldots{}}

\numberwithin{equation}{section}

\begin{document}
\title[Self-adjointness for magnetic operators]{On essential self-adjointness for magnetic Schr\"{o}dinger and Pauli operators
 on the  unit disc in $\IR^2$}
\author[Gh.~Nenciu and I.~Nenciu]{Gheorghe Nenciu and Irina Nenciu}
\address{Gheorghe Nenciu\\
         Institute of Mathematics ``Simion Stoilow''
     of the Romanian Academy\\ 21, Calea Grivi\c tei\\010702-Bucharest, Sector 1\\Romania}
\email{Gheorghe.Nenciu@imar.ro}
\address{Irina Nenciu\\
         Department of Mathematics, Statistics and Computer Science\\
         University of Illinois at Chicago\\
         851 S. Morgan Street\\
         Chicago, IL \textit{and} Institute of Mathematics ``Simion Stoilow''
     of the Romanian Academy\\ 21, Calea Grivi\c tei\\010702-Bucharest, Sector 1\\Romania}
\email{nenciu@uic.edu}

\begin{abstract}
We study the question of magnetic confinement of  quantum particles on the unit disk $\ID$ in $\IR^2$, i.e. we wish to achieve 
 confinement solely by means of the growth of the magnetic field $B(\vec x)$ near the boundary
 of the disk. In the spinless case we show that 
$B(\vec x)\ge \frac{\sqrt 3}{2}\cdot\frac{1}{(1-r)^2}-\frac{1}{\sqrt 3}\frac{1}{(1-r)^2\ln \frac{1}{1-r}},$
for $|\vec x|$ close to 1, insures the confinement  provided we assume that the non-radially symmetric part
 of the magnetic field is not very singular near the boundary. Both constants $\frac{\sqrt 3}{2}$ and $-\frac{1}{\sqrt 3}$
are optimal. This answers, in this context, an open question from \cite{dVT}. 
We also derive growth conditions for radially symmetric magnetic fields which lead to confinement of spin ${1/2}$ particles.
\end{abstract}
\maketitle


\section{Introduction}\label{S:1}

This note is concerned with the problem of  confinement of quantum particles by magnetic fields. At the mathematical level,
 the confinement of quantum particles in a bounded domain $\Omega$ is described by the fact that the corresponding Hamiltonian is 
essentially self-adjoint on $C_0^{\infty}(\Omega)$. The case when the confinement  is due to the presence of an electric field is well
 understood both at the physical and mathematical levels (see  \cite{Bru}, \cite{BMS}, \cite{KSW}, \cite{NenNen} and the references therein):
 a sufficiently fast growth of the electric potential will prevent the particle from reaching the boundary of $\Omega$, leading to confinement.
 Moreover, optimal growth rates are known for the potential close to the boundary $\partial\Omega$, which guarantee essential self-adjointness.

On the contrary, the case in which the confinement is due only to the presence of a magnetic field is much less well-understood; even 
at the physical heuristic level we are not aware of a clear-cut argument justifying confinement. At the mathematical level, it was proved only 
very recently by Colin de Verdi\`ere and Truc \cite{dVT} that, under very general conditions, an inverse square increase of the modulus of the
 magnetic field, $|B(x)|\ge  \frac{C}{\text{dist}(x, \partial \Omega)^2}$, $C>1$, close to the boundary of $\Omega$, leads to confinement 
for spinless particles
(i.e. $H = (-i \nabla -A)^2$ where $A$ is a magnetic potential corresponding to $B$). The main technical ingredient of the proof of confinement in
\cite{dVT} is a lower bound of the quadratic form $h_H$ of $H$ in terms of the magnetic field.

In the rest of this note we shall consider the magnetic confinement problem in the simplest setting  when $\Omega=\ID=\{\vec x\in\IR^2\,|\, |\vec x| < 1\}$ is 
the unit disk in two dimensions. In this setting, the lower bound in \cite{dVT} for $h_H$ is an elementary result saying that, provided $B(\vec x) \ge 0$, 
\begin{equation}\label{glb}
h_H(u,u) \ge  \int_{\Omega} B(\vec x)|u(\vec x)|^2\,dx\,.
\end{equation}
This, together with general results on essential self-adjointness (see \cite{NenNen} and references therein), leads to confinement, as long as close to the boundary of $\Omega$
\begin{equation}\label{C=1}
B(x)\ge  \frac{1}{\text{dist}(x, \partial \Omega)^2}\,.
\end{equation}
As for the optimality of \eqref{C=1}, one can easily give an example (see \cite{dVT}, Theorem 5.8) of a radial magnetic field,
$B(\vec x) \sim \frac{\alpha}{\text{dist}(x, \partial \Omega)^2}$ near the boundary, such that for $\alpha \in (0,\frac{\sqrt 3}{2})$, 
H is not essentially self-adjoint. This raises the question of finding the optimal (i.e. the weakest) increase of 
the magnetic field near the boundary insuring the  essential self-adjointness of $H$. In particular at the level of power like behavior 
the problem left open in \cite{dVT}
 is to find the optimal $C\in [\frac{\sqrt 3}{2},1]$ leading to confinement. Passing to the (most interesting from the 
physical point of view) case of spin 1/2 particles, i.e. when $H$ is replaced by 
$(-i \nabla -A)^2-B$, the problem of confinement is wide open, since \eqref{glb} gives only that $(-i \nabla -A)^2-B \ge 0$, which does not imply confinement,
irrespective of the strength of the magnetic field. The existence of magnetic confinement for spin 1/2 particles is one of the main outcomes of our paper.

In this note we report results about optimal magnetic field increase near the boundary leading to confinement. In the spinless (Schr\"odinger) case, for
$B(\vec x)=B_{rad}(|\vec x|) + B_{1}(\vec x)$ and as long as the non-radially symmetric part of 
the magnetic field, $B_{1}(\vec x)$, is not very singular near the boundary of $\ID$, we prove confinement for
$$ B_{rad}(|\vec x|) \ge \frac{\sqrt 3}{2}\cdot\frac{1}{(1-|\vec x|)^2}-\frac{1}{\sqrt 3}\frac{1}{(1-|\vec x|)^2\ln \frac{1}{1-|\vec x|}}\,.$$
Here both constants in front of the leading and subleading terms are optimal (see Theorem~\ref{PSP} for a 
precise formulation). In particular this settles, for the case at hand, the question left open in 
\cite{dVT}. As for the spin 1/2 (Pauli) case, we prove confinement if the magnetic field is radially symmetric and obeys near $|\vec x|=1$:
$$   \frac{\alpha}{(1-|\vec x|)^2}-\frac{1}{2\alpha}\frac{1}{(1-|\vec x|)^2\ln \frac{1}{1-|\vec x|}} \le B(|\vec x|) 
\le \frac{\beta \alpha}{(1-|\vec x|)^2} $$
where
\begin{equation}\label{alphabeta}
\beta \ge 1, \alpha \ge \frac{\beta +\sqrt{\beta^2 +3}}{2}.
\end{equation}
Notice that, from \eqref{alphabeta}, $\alpha \ge\frac{3}{2}.$
Again the value $\alpha = \frac{3}{2}$ is optimal. By some (tedious) extra work one can add higher order subleading terms of the form
$$
\frac{\text{const.}}{\bigl(1-|x|\bigr)^2 \ln\ln\cdots\ln\frac{1}{1-|x|}}
$$
and determine the corresponding optimal constant.

We wish to comment on the condition that the magnetic field has radial symmetry, which is crucial for our proofs (the non-radially symmetric
 case for spinless particles follows 
from the radially symmetric one by perturbation theory). The point is that, as it stands, the ``global''  lower bound \eqref{glb} seems hard to improve
(if possible at all -- see Remark~4.9 in \cite{dVT}); as already mentioned, this leads to $C \ge 1$. The radial symmetry allows for  partial wave decomposition, and thus reduces the 
 essential self-adjointness problem for the whole operator to the 
one for {\em each} partial wave sector (indexed by the magnetic quantum number $m\in {\mathbb Z}$). 

We would like to stress
that the point of this reduction is not the fact that one ends up with a collection of 1 dimensional problems, for which one uses Weyl limit point criteria; 
the present 
day criteria for essential self-adjointness are almost as powerful in the multi dimensional case as in 1 dimension, 
see  \cite{Bru}, \cite{BMS}, \cite{KSW}, \cite{NenNen} and the references therein (actually the limit point criterion we use is a particular
 of case of the 
multi dimensional result in \cite{NenNen}). What we gain from this decomposition is rather the fact that we are left with the problem of 
proving appropriate lower 
bounds for the effective one dimensional potential in each sector. It turns out that this can be done (see Lemma~\ref{Sest} below),
but note that these bounds  are not uniform in $m$, in
 the sense that they are valid only for $|\vec x| > r_m$ with $\lim_{m \rightarrow \infty}r_m =1.$ 

At the technical level, aside from the results in \cite{NenNen}, the main ingredient is the fact that the 
formula giving the  magnetic vector potential in 
the transversal gauge (see Lemma \ref{L:3.0.1}) allows a ``nice'' transfer of the growth conditions from the magnetic field 
to the corresponding magnetic potential entering the Schr\"odinger or Pauli operators (see Lemma~\ref{Sest}).

The paper is organized as follows: Sections~\ref{S:3} and \ref{S:4} contain our main results and their proofs, respectively. 
The two appendices have very different character:
 Appendix~\ref{A:1} contains  a 1-dimensional version of the essential self-adjointness  
criterion from \cite{NenNen} expressed as a new, integral, limit point criterion. 
Aside from its use in the proof of the main result, this criterion might be of interest in itself as a refinement or easier-to-use version of many of the known limit point criteria (see e.g. Theorem~X.10 in \cite{ReeSim}, Theorem~1 in \cite{H}). 
Finally, Appendix~\ref{S:G} is included for the reader's convenience, as it contains some of the known properties of 
the transversal gauge which we use in our proofs.

\section{Set-up of the problem and results}\label{S:3}

As already mentioned in the introduction, we will restrict our attention to the case when
\begin{equation}\label{E:3.1}
\Omega=\ID=\{\vec x=(x_1,x_2)\in \IR^2\,|\, x_1^2+x_2^2<1\}\subset \IR^2\,,
\end{equation}
and we will consider a magnetic Schr\"{o}dinger operator, 
\begin{equation}\label{E:3.2}
H^S=\bigl(-i\nabla-A\bigr)^2+q\,,
\end{equation}
and the associated Pauli operator 
\begin{equation}\label{E:3.2.1}
\bigl(-i\nabla-A\bigr)^2+q \pm B
\end{equation}
which appears as the nonrelativistic limit of the corresponding Dirac operator \cite{Th}.
We assume throughout the paper that 
\begin{equation}\label{cond1}
B\ge 0\qquad\text{and}\qquad B\in C^1(\ID) 
\end{equation}
This in particular implies that we need only discuss, for the Pauli operator,  the nontrivial case, 
\begin{equation}\label{E:3.2.2}
H^P=\bigl(-i\nabla-A\bigr)^2+q-B\,.
\end{equation}

An important ingredient in proving essential self-adjointness for our examples is the choice of the \emph{transversal} 
(or \emph{Poincar\'e}) gauge for the magnetic vector potential. So throughout the paper, for a given magnetic field $B(\vec x)$, $A(\vec x)$ 
denotes the corresponding magnetic potential in the transversal gauge. From our results, essential self-adjointness follows for all other gauges 
related to the transversal one by smooth gauge transformations (see e.g. Proposition 2.13 in \cite{dVT}). 
For the definition of the transversal gauge and a few properties used in our proofs, see Lemmas \ref{L:3.0.1} and 
\ref{L:3.0.2} below; more properties can be found, e.g., in \cite{Th} or \cite{N}. 

We are interested in finding conditions on the magnetic field $B(\vec x)$ near $|\vec x|=1$ guaranteeing 
the essential self-adjointness of $H^S$ and $H^P$ in the case when the scalar potential,  $q(\vec x)$, vanishes (or is uniformly bounded) near $|\vec x|=1$. Since we shall make heavy use of polar coordinates, $\vec x =(r,\theta)$, in order not to obscure the main ideas with irrelevant technicalities related with the singularity of the transformation from rectangular to polar coordinates near the origin, we shall consider the essential self-adjointness problem for $H^S$ and $H^P$ on $C_0^{\infty}(\Omega)$
where 
\begin{equation}\label{E:omega}
\Omega = \{ 0<|\vec x|<1\},
\end{equation} 
and
\begin{equation}\label{q}
q(\vec x) = \frac{1}{|\vec x|^2}\,.
\end{equation}
We would like to emphasize the fact that, since $q(\vec x)$ as given by \eqref{q} assures the self-adjointness at $0$ (see \cite{ReeSim}), for a given magnetic field the essential self-adjointness of $H^S$ and $H^P$ on  $C_0^{\infty}(\{|\vec x|<1\})$ with $q=0$ is {\em equivalent} to the essential self-adjointness of $H^S$ and $H^P$ (respectively) on  $C_0^{\infty}(\{0<|\vec x|<1\})$ with $q$ as given by \eqref{q}.

We are now in the position to state the main result of this note. In what follows, $r$ and $\theta$ are the polar coordinates of $\vec x$.
\begin{theorem}\label{PSP}
 (i) Consider the Schr\"odinger operator
$$
H^S=\bigl(-i\nabla-A\bigr)^2+\frac{1}{r^2}
$$
defined on $\mathcal D(H^S)=C_0^{\infty}(\{0<|\vec x|<1\})$, where 
\begin{equation}\label{BS}
B(\vec x) = B_{rad}(r) + B_{1}(\vec x)
\end{equation}
with
\begin{equation}\label{BSr}
 B_{rad}(r) \ge \frac{\sqrt 3}{2}\cdot\frac{1}{(1-r)^2}-\frac{1}{\sqrt 3}\frac{1}{(1-r)^2\ln \frac{1}{1-r}},
\end{equation}
for $r$ close to 1, and
\begin{equation}\label{BS1}
 \int_0^1\!\bigl|B_1(r,\theta)\bigr|\,dr + \int_0^1\!\bigl|\frac{\partial}{d \theta}B_1(r,\theta)\bigr|\,dr < \infty
\end{equation}
uniformly in $\theta \in [0, 2\pi)$.
Then $H^S$ is essentially self-adjoint.

(ii) Consider the Pauli operator
$$
H^P=\bigl(-i\nabla-A\bigr)^2+\frac{1}{r^2} -B(\vec x)
$$
defined on ${\mathcal D}(H^P)=C_0^{\infty}(\{0<|\vec x|<1\})$, where 
\begin{equation}\label{BP}
B(\vec x) = B_{rad}(|\vec x|). 
\end{equation}
Suppose
\begin{equation}
  \beta \ge 1, \alpha \ge \frac{\beta +\sqrt{\beta^2 +3}}{2},
\end{equation}
and
\begin{equation}\label{BPr}
\frac{\alpha}{(1-r)^2}-\frac{1}{2\alpha}\frac{1}{(1-r)^2\ln \frac{1}{1-r}} \le B_{rad}(r) \le \frac{\beta \alpha}{(1-r)^2}, 
\end{equation}
for $r$ close to 1.
Then $H^P$ is essentially self-adjoint. 

(iii) For any $c<\frac{\sqrt 3}{2}$ and $d>\frac{1}{\sqrt 3}$ one can find magnetic fields $ B(r)$ satisfying either
$B(r) \ge c\frac{1}{(1-r)^2}$ or $B(r) \ge  \frac{\sqrt 3}{2}\cdot\frac{1}{(1-r)^2}-d\frac{1}{(1-r)^2\ln \frac{1}{1-r}}$ for which $H^S$
is not essentially self-adjoint.

 For any $c<\frac{ 3}{2}$ and $d>\frac{1}{ 3}$ one can find magnetic fields $ B(r)$ satisfying for  $r$ close to 1 either
$B(r) \ge c\frac{1}{(1-r)^2}$ or $B(r) \ge  \frac{ 3}{2}\cdot\frac{1}{(1-r)^2}-d\frac{1}{(1-r)^2\ln \frac{1}{1-r}}$ for which $H^P$
is not essentially self-adjoint.
\end{theorem}

\section{Proofs}\label{S:4}

\subsection{Proof of Theorem~\ref{PSP}(i) in the radially symmetric case}\label{Ss1}
We will start by providing growth conditions close to the boundary for the magnetic field in the radially symmetric
 Schr\"odinger case. That means that, in the transversal gauge, the magnetic potential $A$ has the form given by Lemmas~\ref{L:3.0.1} and~\ref{L:3.0.2}:
\begin{equation}\label{E:3.3}
A(r,\theta)=ra(r)\,\bigl(-\sin\theta\quad\cos\theta\bigr)\,,
\end{equation}
with
\begin{equation*}
a(r)=\int_0^1\!\! tB_\text{rad}(tr)\,dt 
\end{equation*}
where $(r,\theta)$ are, as before, the polar coordinates corresponding to the rectangular coordinates $\vec x=(x_1,x_2)$.

We argue now following \cite{ReeSim}, \cite{Cor} (the argument is quite standard but we include it for completeness) that 
the essential self-adjointness of $H^S$ is equivalent with the essential self-adjointness of
\begin{equation}\label{E:3.13}
\tilde H^S_{m}\equiv  -\frac{d^2}{dr^2}+\frac{3}{4r^2}+\left(ra(r)-\frac{m}{r}\right)^2\quad\text{in}\quad L^2((0,1),dr)\,,
\end{equation}
for all $m\in\ZZ$, defined on $C_0^{\infty}((0,1))$.
Indeed, let us note first that, according to Lemma~\ref{L:3.1}, the Schr\"odinger operator written in polar coordinates takes the form 
\begin{equation}\label{E:3.4}
H^S=-\frac1r\,\frac{\partial}{\partial r}\,r\,\frac{\partial}{\partial r}+
\left(-\frac{i}{r}\,\frac{\partial}{\partial\theta}-ra(r)\right)^2+\frac{1}{r^2}
\end{equation}
We now split the operator $H^S$ according to partial waves. For
\begin{equation}\label{E:3.8}
\psi_m(r,\theta)=\varphi(r)\cdot e^{im\theta}\,,
\end{equation} 
we obtain, using \eqref{E:3.4}, that
\begin{equation*}\label{E:3.9}
(H^S \psi_m)(r,\theta) =\left[-\frac{1}{r}\frac{\partial}{\partial r}r\frac{\partial}{\partial r}+q(r)+r^2a(r)^2-
2ma(r)+\frac{m^2}{r^2}\right]\varphi(r)\cdot e^{im\theta}\,
\end{equation*}
acting on the space $L^2((0,1)\times S^1,rdrd\theta)$.
Then by a standard argument (see e.g. \cite{ReeSim}, Appendix XI, Example 4) the essential self-adjointness of $H^S$ 
is equivalent with the essential self-adjointness for all $m\in \mathbb Z$ of
\begin{equation*}\label{E:3.10.1}
H^S_m=-\frac{1}{r}\frac{d}{dr}r\frac{d}{dr}+\frac{1}{r^2}+r^2a(r)^2-2ma(r)+\frac{m^2}{r^2}
\end{equation*}
defined on $C_0^{\infty}((0,1))$ in $L^2((0,1),rdr)$.
In other words, we are interested, for each $m\in\ZZ$, in the operator
$H^S_m$ on the space $L^2((0,1),rdr)$, and more precisely, we want to understand its essential self-adjoiness properties for $r$ close to 1.

Using the general notation and set-up for  Sturm-Liouville transformations from \cite{Cor} (for this particular 
case see also \cite{ReeSim},\cite{dVT}), we define the unitary operator of multiplication with $\gamma (r) =r^{-1/2}$, 
\begin{equation}\label{E:3.11}
\Gamma\,:\, L^2((0,1),dr) \rightarrow L^2((0,1),rdr)\qquad\big(\Gamma \phi\big)(r)=r^{-1/2}\phi (r)\,.
\end{equation}
Then we know from the general theory that
\begin{equation}\label{E:3.12}
\Gamma^{-1}H^S_{m}\Gamma=-\Delta+\tilde q_m,\quad\text{with}\quad \tilde q_m^S=\frac{H^S_{m}\gamma}{\gamma}\,.
\end{equation}
In our situation, we see from \eqref{E:3.10.1} and \eqref{E:3.12} that
\begin{equation}\label{E:3.13a}
\tilde H^S_{m} = \Gamma^{-1}H^S_{m}\Gamma .
\end{equation}
The main point of \eqref{E:3.13a} is that $C_0^{\infty}((0,1))$ is invariant under $\Gamma$, so $\tilde H^S_{m}$
is essentially self-adjoint on $C_0^{\infty}((0,1))$ if and only if $H^S_{m}$ is essentially self-adjoint on $C_0^{\infty}((0,1))$.
So now we need to look at which growth conditions on $B$ lead to situations for which 1 is a limit-point of 
$$
\tilde H^S_{m}=-\frac{d^2}{dr^2}+\frac{3}{4r^2}+\frac{1}{r^2}\left(r^2a(r)-m\right)^2
$$
with domain $\mathcal D(\tilde H^S_m)=C_0^\infty\bigl((0,1)\bigr)$, $m\in\ZZ$. Recall that this depends only on the growth 
rate close to 1 of the potential
\begin{equation}\label{E:3.14}
\tilde q_m^S(r)=r^2a(r)^2-2ma(r)+\frac{4m^2+3}{4r^2}\,.
\end{equation}

To describe the growth condition on the magnetic field, we consider the ``critical'' magnetic field
\begin{equation}\label{E:3.15}
B_c(r)=
\begin{cases}
\frac{\sqrt 3}{2}\frac{1}{(1-r)^2}-\frac{1}{\sqrt 3}\frac{1}{(1-r)^2\ln \frac{1}{1-r}}, & \text{for $1-r \le e^{-4}$,}\\
0, & \text{ otherwise}
\end{cases}
\end{equation}
\begin{prop}\label{L:3.2}
Suppose that 
\begin{equation}\label{E:3.16}
B(r)\ge B_c(r)\,.
\end{equation}
Then $\tilde H^S_m$ is essentially self-adjoint for all $m\in\ZZ$.
\end{prop}

\begin{proof}
According to Weyl theory (see Theorem~X.7 in \cite{ReeSim}), we have to verify that $\tilde H^S_m$ is limit point at 0 and 1. Since
\begin{equation}\label{E:3.17}
\tilde q_m^S(r)=\frac{3}{4r^2}+\frac{1}{r^2}\bigl(r^2a(r)-m\bigr)^2\ge \frac{3}{4r^2}
\end{equation}
it follows from classical results (e.g., Theorem~X.10 in \cite{ReeSim}) that $\tilde H^S_m$ is limit point at 0.

Concerning the situation at 1, 
 the needed growth rate  close to 1 of $\tilde q_m^S(r)$
is provided by the following technical lemma.
\begin{lemma}\label{Sest}
 For each $m \in \mathbb{Z} $ there exist $r_m <1$ such that for $r_m <r<1$:
\begin{equation}\label{condA1}
 \frac{1}{4(1-r)^2}+ \tilde q_m^S(r)\ge 
\left(\frac{1}{1-r}-\frac{1}{2(1-r)\ln\frac{1}{1-r}}-\frac{4}{(1-r)\ln^2\frac{1}{1-r}}
\right)^2
\end{equation}
and
\begin{equation}\label{condA2}
 \frac{1}{2(1-r)} \le \frac{1}{1-r}-\frac{1}{2(1-r)\ln\frac{1}{1-r}}-\frac{4}{(1-r)\ln^2\frac{1}{1-r}}
\end{equation}
\end{lemma}
Taking Lemma \ref{Sest} for granted one can finish the proof of Proposition \ref{L:3.2}. 
Indeed, choosing the function $G$ to be (for small enough $t$) of the form
\begin{equation}\label{E:3.22}
G(t)=\ln t+ \frac12\ln\ln\frac1t + \int_t\!\,\frac{4}{(1-u)\ln^2\frac{1}{1-u}}du\,,
\end{equation}
one can apply directly Lemma~\ref{L:A.1}, which gives that $\tilde H^S_m$ is limit point at 1.
\end{proof}

We turn now to the proof of Lemma \ref{Sest}. 
\begin{proof}[Proof of Lemma~\ref{Sest}]
Let $1-r_0=e^{-4}$. A finite number of constants appearing 
during the proof will be denoted by the same letter $C$. From  \eqref{E:3.2.6}, \eqref{E:3.15} and \eqref{E:3.16} for $r>r_0$
\begin{equation}\label{BrBc}
\begin{aligned}
r^2a(r)=& r^2\int_0^1\!\! tB_\text{rad}(tr)\,dt=\int_{0}^r\!\! uB_\text{rad}(u)\,du \ge \int_{r_0}^r\!\! uB_c(u)\,du \\
& = \frac{\sqrt{3}}{2} \int_{r_0}^r \!\!\frac{u}{(1-u)^2}\,du-\frac{1}{\sqrt{3}} \int_{r_0}^r\!\!\frac{u}{(1-u)^2\ln\frac{1}{1-u}}\,du.
\end{aligned}
\end{equation}
The first term on the r.h.s. of \eqref{BrBc} gives 
\begin{equation}\label{main}
 \frac{\sqrt{3}}{2} \left(\frac{1}{1-r}-\ln\frac{1}{1-r}\right) +C.
\end{equation}
We estimate now the second term on the r.h.s. of \eqref{BrBc}. Integration by parts gives
\begin{equation}\label{Bcln1}
\begin{aligned}
\int_{r_0}^r\!\!\frac{u}{(1-u)^2\ln\frac{1}{1-u}}&\,du \le \int_{r_0}^r\!\!\frac{1}{(1-u)^2\ln\frac{1}{1-u}}\,du \\
&= \frac{1}{(1-r)\ln\frac{1}{1-r}}+\int_{r_0}^r\!\!\frac{1}{(1-u)^2\ln^2\frac{1}{1-u}}\,du +C.
\end{aligned}
\end{equation}
Integrating once again by parts and taking into account that for $r>r_0$, $1-r <e^{-4}$:
\begin{equation*}
\begin{aligned}
 \int_{r_0}^r\!\!\frac{1}{(1-u)^2\ln^2\frac{1}{1-u}}\,du& =
\frac{1}{(1-r)\ln^2\frac{1}{1-r}} +2\int_{r_0}^r\!\!\frac{1}{(1-u)^2\ln^3\frac{1}{1-u}}\,du +C \\
&\le \frac{1}{(1-r)\ln^2\frac{1}{1-r}} +\frac{1}{2}\int_{r_0}^r\!\!\frac{1}{(1-u)^2\ln^2\frac{1}{1-u}}\,du +C
\end{aligned} 
\end{equation*}
which gives
\begin{equation}\label{ln2}
\int_{r_0}^r\!\!\frac{1}{(1-u)^2\ln^2\frac{1}{1-u}}\,du \le  \frac{2}{(1-r)\ln^2\frac{1}{1-r}} +C.
\end{equation}
From \eqref{Bcln1} and \eqref{ln2}
\begin{equation}\label{sublead} 
\int_{r_0}^r\!\!\frac{u}{(1-u)^2\ln\frac{1}{1-u}}\,du \le \frac{1}{(1-r)\ln\frac{1}{1-r}}+\frac{2}{(1-r)\ln^2\frac{1}{1-r}} +C.
\end{equation}
Putting together \eqref{BrBc}, \eqref{main} and \eqref{sublead} one obtains
\begin{equation}\label{argeq}
\begin{aligned}
 r^2a(r) &\ge  \frac{\sqrt{3}}{2} \frac{1}{(1-r)}-\frac{1}{\sqrt{3}}\frac{1}{(1-r)\ln\frac{1}{1-r}}\\
&\qquad-\frac{2}{(1-r)\ln^2\frac{1}{1-r}}-\frac{\sqrt{3}}{2}\ln\frac{1}{1-r}+C\,.\\
\end{aligned}
\end{equation}
Choose $r_m \ge r_0$ such that
\begin{equation}\label{rm}
 \frac{1}{(1-r)\ln^2\frac{1}{1-r}} \ge \frac{\sqrt{3}}{2}\ln\frac{1}{1-r} +m -C.
\end{equation}
Then for $r>r_m$, from \eqref{argeq}
\begin{equation}\label{ar-m}
  r^2a(r) -m \ge \frac{\sqrt{3}}{2} \frac{1}{(1-r)}-\frac{1}{\sqrt{3}}\frac{1}{(1-r)\ln\frac{1}{1-r}} 
-\frac{3}{(1-r)\ln^2\frac{1}{1-r}} \ge 0
\end{equation}
Since (see \eqref{E:3.17}) $\tilde q_m^S(r)\ge (r^2a(r) -m)^2 $, from \eqref{ar-m} one can
 check by direct computation (use $\ln\frac{1}{1-r}\ge 4$) that \eqref{condA1} and \eqref{condA2} hold true.
 \end{proof}

\subsection{Proof of Theorem~\ref{PSP}(i) in the nonradial case}\label{Ss2}
Write (see \eqref{E:3.2.4} and \eqref{E:3.2.6})
\begin{equation}\label{aa1}
a(r,\theta)=a_{\text{rad}}(r)+a_1(r,\theta)
\end{equation}
where
\begin{equation}\label{a1}
a_1(r,\theta)=\int_0^1\! tB_1(t\vec x)\,dt,\quad a_{\text{rad}}(r)=\int_0^1\!tB_{\text{rad}}(tr)\,dt\,.
\end{equation}
From \eqref{a1} and \eqref{BS1} it follows that $a_1$ and $\frac{\partial a_1}{\partial\theta}$ are uniformly bounded:
\begin{equation}\label{A1}
\sup_{|\vec x|<1}\left\{\big|a_1(r,\theta)\big|+\Big|\frac{\partial a_1}{\partial\theta}(r,\theta)\Big|\right\}\le A_1<\infty\,.
\end{equation}
Here recall (see \eqref{cond1}) that we always assume that our magnetic fields, in particular $B_1$, are $C^1$-smooth 
on the whole unit disk, including at 0. This is needed to justify the uniformity of the bound \eqref{A1} as $r\to 0$.

Plugging \eqref{aa1} into \eqref{E:3.2.7} and expanding, one obtains
\begin{equation}\label{HS1}
H^S=H^S_{\text{rad}}+H^S_1\,,
\end{equation}
where
\begin{equation}\label{Hrad}
H^S_\text{rad}=-\frac1r\,\frac{\partial}{\partial r}\,r\,\frac{\partial}{\partial r}+P_{\theta,\text{rad}}^2+\frac{1}{r^2}
\end{equation}
and
\begin{equation}\label{H1}
\begin{aligned}
H^S_1
&=ra_1(r,\theta)P_{\theta,\text{rad}}+P_{\theta,\text{rad}}ra_1(r,\theta)+r^2a_1(r,\theta)^2\\
&=2ra_1(r,\theta)P_{\theta,\text{rad}}-i\frac{\partial a_1}{\partial\theta}(r,\theta)+r^2a_1(r,\theta)^2\,,
\end{aligned}
\end{equation}
with
\begin{equation}\label{Prad}
P_{\theta,\text{rad}}=-\frac{i}{r}\,\frac{\partial}{\partial\theta}+ra_\text{rad}(r)\,.
\end{equation}

Notice that both $H^S_\text{rad}$ and $H^S_1$ are symmetric on $C^\infty_0(\{0<|\vec x|<1\})$ and, from subsection~\ref{Ss1}, $H^S_\text{rad}$ 
is essentially self-adjoint. We show now that $H^S_1$ is relatively bounded with respect to $H^S_\text{rad}$, which will complete this part of the proof.

By \eqref{A1}, the last two terms on the right-hand side of \eqref{H1} are bounded, and so we need only consider the first term. Let 
$$
\varphi\in C_0^\infty(\{0<|\vec x|<1\}) \,.
$$
Then
\begin{equation*}
\begin{aligned}
\bigl\|2ra_1(r,\theta)P_{\theta,\text{rad}}\varphi\bigr\|^2 
&\le 2A_1^2\langle\varphi,P^2_{\theta,\text{rad}}\varphi\rangle\le 2A_1^2\langle\varphi, H^S_\text{rad}\varphi\rangle\\
&\le 2A_1^2\|\varphi\|\cdot\|H^S_\text{rad}\varphi\|\le \frac{A_1^2}{d^2}\|\varphi\|^2+A_1^2d^2\|H^S_\text{rad}\varphi\|^2\\
&\le A_1^2\left(\frac{\|\varphi\|}{d}+d\|H^S_\text{rad}\varphi\|\right)^2\,,
\end{aligned}
\end{equation*}
where we used the general fact that $2ab\le a^2+b^2$. Putting all together yields
\begin{equation*}
\|H^S_1\varphi\|\le A_1d\|H^S_\text{rad}\varphi\|+\left(\frac{A_1}{d}+A_1+A_1^2\right)\|\varphi\|\,,
\end{equation*}
which leads to the needed bound when $d$ is chosen small enough. The essential self-adjointness of $H^S$ then follows from the stability of
 essential self-adjointness against relatively bounded perturbations (see e.g. \cite{Ka}, \cite{ReeSim}).

\subsection{Proof of Theorem~\ref{PSP}(ii)}\label{Ss3}

By the same argument as in the radially symmetric Schr\"odinger case, one is reduced to the proof of essential self-adjointness of
\begin{equation*}
\tilde H^P_m=-\frac{d^2}{dr^2}+\frac{3}{4r^2}+\frac{1}{r^2}\bigl(r^2a_\text{rad}(r)-m\bigr)^2-B_\text{rad}(r)
=-\frac{d^2}{dr^2}+\tilde q^P_m(r)\,.
\end{equation*}

Let $r_{\alpha}$ defined by $1-r_{\alpha}=e^{-2(\alpha +1)}$.
Defining
\begin{equation}\label{BcP}
B_{c,\alpha}(r)=
\begin{cases}
\frac{\alpha}{(1-r)^2}-\frac{1}{2\alpha}\frac{1}{(1-r)^2\ln \frac{1}{1-r}}, & \text{for $1-r \le e^{-2(\alpha +1)}$,}\\
0, & \text{ otherwise}
\end{cases}
\end{equation}
and mimicking closely the proof of Lemma \ref{Sest} one obtains 
\begin{equation}\label{Pargeq}
 r^2a_\text{rad}(r) \ge   \frac{\alpha}{(1-r)}-\frac{1}{2\alpha}\frac{1}{(1-r)\ln\frac{1}{1-r}}
-\frac{1}{(1-r)\ln^2\frac{1}{1-r}}-\alpha\ln\frac{1}{1-r}+C(\alpha).
\end{equation}
Choose $r_{m,\alpha} \ge r_\alpha$ such that
\begin{equation}\label{Prm}
 \frac{1}{(1-r)\ln^2\frac{1}{1-r}} \ge \frac{\sqrt{3}}{2}\ln\frac{1}{1-r} +m -C(\alpha).
\end{equation}
Then for $r>r_{m,\alpha}$ , from \eqref{Pargeq}
\begin{equation}\label{Par-m}
  r^2a_\text{rad}(r) -m \ge \frac{\alpha}{(1-r)}-\frac{1}{2\alpha}\frac{1}{(1-r)\ln\frac{1}{1-r}}
-\frac{2}{(1-r)\ln^2\frac{1}{1-r}}
 \ge 0.
\end{equation}
 From \eqref{BPr},
\begin{equation*}
\frac{1}{4(1-r)^2}+\tilde q^P_m(r)\ge\bigl(r^2a_\text{rad}(r)-m\bigr)^2-\frac{\alpha\beta-\frac14}{(1-r)^2}.
\end{equation*}
Then from \eqref{Par-m} one can again verify directly (notice that from \eqref{alphabeta}, 
$\alpha \ge \frac{3}{2} $ and $\alpha^2-\alpha\beta+\frac14\ge1$, and that $\ln\frac{1}{1-r_\alpha}=2(\alpha+1)$) that
\begin{equation}\label{condAP1}
 \frac{1}{4(1-r)^2}+\tilde q^P_m(r)\ge 
\left(\frac{1}{1-r}-\frac{1}{2(1-r)\ln\frac{1}{1-r}}-\frac{2(\alpha +1)}{(1-r)\ln^2\frac{1}{1-r}}
\right)^2
\end{equation}
and
\begin{equation}\label{condAP2}
 \frac{1}{2(1-r)} \le \frac{1}{1-r}-\frac{1}{2(1-r)\ln\frac{1}{1-r}}-\frac{2(\alpha+1)}{(1-r)\ln^2\frac{1}{1-r}}.
\end{equation}
 A direct application of Lemma~\ref{L:A.1} with
\begin{equation}\label{GP}
G(t)=\ln t+ \frac12\ln\ln\frac1t + \int_t\!\,\frac{2(\alpha+1)}{(1-u)\ln^2\frac{1}{1-u}}\,du\,,
\end{equation}
completes the proof of Theorem~\ref{PSP}(ii).
\subsection{Proof of Theorem~\ref{PSP}(iii)}

 For $B(r) \ge c\frac{1}{(1-r)^2}$  in the Schr\"odinger case  see Theorem 5.8 in \cite{dVT}. 
For the optimality of the constant in front of the subleading term choose
\begin{equation}\label{aopt}
a(r)=
\begin{cases}
\frac{\sqrt 3}{2}\frac{1}{1-r}-\frac12(d+\frac{1}{\sqrt 3})\frac{1}{(1-r)\ln \frac{1}{1-r}}-e
(\frac{1}{\sqrt 3}-\frac{d}{2}), & \text{for $1-r \le \frac{1}{e}$,}\\
0, & \text{ otherwise}
\end{cases}
\end{equation}
and verify that the corresponding magnetic field  has the right behavior as $r\rightarrow 1$.
At the same time for $r$ sufficiently close to one from \eqref{aopt} and \eqref{E:3.14}
\begin{equation}
\tilde q_0^S (r) = r^2a(r)^2+\frac{3}{4r^2}\ge \frac{3}{4}\frac{1}{(1-r)^2}-\frac{d\sqrt 3+1}{2}\cdot\frac{1}{(1-r)^2\ln \frac{1}{1-r}},
\end{equation}
and since $d\sqrt 3 >1$ one can apply Theorem 3 in \cite{NenNen} to $\tilde H_m^S$. The proof for the Pauli case is similar.

\appendix

\section{Background} \label{A:1}

In this appendix, we give the particular case we need of the main theorem in \cite{NenNen}, in a form best adapted to its application in this paper.

\begin{lemma}\label{L:A.1}
Let
\begin{equation}\label{E:A1}
\bigl(H\varphi\bigr)(x)=-\varphi''(x)+V(x)\varphi(x)\qquad\text{on}\quad (0,1)
\end{equation}
with $V$ a continuous potential. Assume that $V=V_1+V_2$, with $V_2$ uniformly bounded and
\begin{equation}\label{E:A2}
V_1(x)+\frac{1}{4(1-x)^2}\ge \left(G'(1-x)\right)^2\qquad\text{for}\quad x\in(1/2,1)
\end{equation}
with $G\,:\,(0,1/2)\rightarrow \IR$ differentiable and satisfying:\\
i. There exists $0<d_0<1/2$ such that 
$$
0\le G'(t)\le\frac1t\,\,\,\text{for}\,\,\, t\in(0,d_0)\,\,\,\text{and}\,\,\, G'(t)=0\,\,\,\text{for}\,\,\, t\ge d_0\,;
$$
ii. For any $\rho_0\le d_0/2$, 
\begin{equation}\label{E:A3}
\sum_{n=1}^\infty 4^{-n} e^{-2G(2^{-n}\rho_0)}=\infty\,.
\end{equation}
Then $H$ is limit-point at 1.
\end{lemma}

\begin{proof}
Let
\begin{equation*}
\tilde V(x)=\begin{cases}
V(x), & \text{for}\,\,\, x\ge\frac12;\\
V(1-x), &\text{for}\,\,\, x<\frac12.
\end{cases}
\end{equation*}
Then the conditions of Theorem~1 in \cite{NenNen} are 
fulfilled for $\tilde H=-\frac{d^2}{dx^2}+\tilde V$ on $\Omega=(0,1)$, $\mathcal D(\tilde H)=C_0^\infty(\Omega)$, so that $\tilde H$ is 
essentially self-adjoint. Thus, by Theorem~X.7 in \cite{ReeSim} it must be limit-point at 1.
\end{proof}

Note that the crucial growth condition for the potential near $x=1$ is \eqref{E:A3}, but that looks somewhat unfamiliar. In fact, it is 
equivalent to an integral type condition, at least in the case where we replace $G'(t)\ge0$ by $G'(t)\ge1/(2t)$. 
The following integral limit-point criterion, which is of interest in itself as a refinement or easier-to-use version of many of the known limit-point criteria 
(see, e.g., Theorem~X.10 in \cite{ReeSim}), is a direct consequence of Lemma~\ref{L:A.1}:
\begin{prop}\label{P:A.2}
Let
\begin{equation*}
\bigl(H\varphi\bigr)(x)=-\varphi''(x)+V(x)\varphi(x)\qquad\text{on}\quad (0,1)
\end{equation*}
with $V$ a continuous potential. Assume that $V=V_1+V_2$, with $V_2$ uniformly bounded and
\begin{equation*}
V_1(x)+\frac{1}{4(1-x)^2}\ge \bigl(G'(1-x)\bigr)^2\qquad\text{for}\quad x\in(1/2,1)
\end{equation*}
with $G\,:\,(0,1/2)\rightarrow \IR$ differentiable and satisfying:\\
i. There exists $0<d_0<1/2$ such that 
\begin{equation}\label{E:A4}
\frac{1}{2t}\le G'(t)\le\frac1t\,\,\,\text{for}\,\,\, t\in(0,d_0),\,\,\, G'(t)=0\,\,\,\text{for}\,\,\, t\in (2d_0,\tfrac12)\,;
\end{equation}
ii. 
\begin{equation}\label{E:A5}
\lim_{\varepsilon\to0+}\int_\varepsilon\! te^{-2G(t)}\,dt=\infty\,.
\end{equation}
Then $H$ is limit-point at 1.
\end{prop}

\begin{proof}
It follows from \eqref{E:A4} that
\begin{equation}\label{E:A6}
\frac{d}{dt}\bigl(te^{-2G(t)}\bigr)\le 0\,.
\end{equation}
Now let $\rho_0\le d_0/2$, as in the hypothesis of Lemma~\ref{L:A.1}, and denote $t_n=2^{-n} \rho_0$. Then the sum in \eqref{E:A3} equals
\begin{equation*}
\frac{1}{\rho_0^2}\sum_{n=1}^\infty t_n e^{-2G(t_n)} \bigl(t_{n-1}-t_n\bigr)
=\frac{2}{\rho_0^2}\sum_{n=1}^\infty t_n e^{-2G(t_n)} \bigl(t_{n}-t_{n+1}\bigr)\,.
\end{equation*}
Together with \eqref{E:A6}, this implies
\begin{equation*}
\frac{1}{\rho_0^2}\int_0^{\rho_0}\! te^{-2G(t)}\,dt\le \sum_{n=1}^\infty 4^{-n}e^{-2G(2^{-n}\rho_0)}
\le \frac{2}{\rho_0^2}\int_0^{\frac{\rho_0}{2}}\! te^{-2G(t)}\,dt\,,
\end{equation*}
showing in particular that \eqref{E:A5} implies \eqref{E:A3}, and completing the proof.
\end{proof}

The simplest choice for $G$ near $t=0$ is $G(t)=\ln t$, and it leads to the result of Theorem~X.10 in \cite{ReeSim}. 
The choice used in the proof of Theorem~\ref{PSP} is of the form
\begin{equation}\label{Glnlnf}
G(t)=\ln t+\frac12\ln\ln\frac1t+\int_tf(u)du
\end{equation}
with $f(u) \ge 0$, $\lim_{u\rightarrow 0}uf(u)=0$ and $\lim_{t\rightarrow 0}\int_tf(u)du < \infty.$
We send the reader to \cite{NenNen} for more examples and a discussion of optimality. The case \eqref{Glnlnf} is not covered e.g. by Theorem~X.10 in \cite{ReeSim} or Lemma~3.11 in \cite{GZ}, nor by the particularization to the 1-dimensional case of Theorem 6.2 in \cite{Bru} or the Main Theorem iii. in \cite{MMc}. Concerning Hinton's Theorem (Theorem 1 in \cite{H}), while the choice $\eta(x)=x^{-1/2}$, given there for the case at hand, also does not cover \eqref{Glnlnf}, one can show that a better choice, $\eta(x)=\frac{1}{x^{1/2}\ln(1/x)^{1/2}}$ near $x=0$, does the job.

We close this appendix with a remark: in order to avoid technicalities, we have imposed smoothness conditions on $V$ and $G$ which are 
stronger than necessary. For example, the differentiability condition for $G$ can be relaxed.
 In fact, as can be seen in the classical Agmon paper~\cite{Agm} (see also~\cite{dVT}), it is sufficient to require that $G$ be 
Lipschitz continuous and the corresponding inequality be understood in an almost everywhere sense.

\section{A few facts about the transversal gauge}\label{S:G}

An important step in proving essential self-adjointness for our examples is the choice of the \emph{transversal} 
(or \emph{Poincar\'e}) gauge for the magnetic vector potential. The definition and properties of the transversal gauge used in 
this paper are well-known (see, e.g., \cite{Th}), and we provide the proofs for the reader's convenience.

More precisely, we have the following:
\begin{lemma}\label{L:3.0.1}
Let 
\begin{equation}\label{E:3.2.2a}
A(\vec x)=\int_0^1\!\!\! t\vec B(t\vec x)\,dt \wedge \vec x=\left(-x_2\int_0^1\!\! tB(t\vec x)\,dt,\quad x_1\int_0^1\!\! tB(t\vec x)\,dt\right)
\end{equation}
Then $A$ is a vector potential associated to the magnetic field $B$ and
\begin{equation}\label{E:3.2.3}
A(\vec x) \perp \vec x,\qquad |A(\vec x)|=|\vec x|\cdot \int_0^1\!\! tB(t\vec x)\,dt\,.
\end{equation}
\end{lemma}

\begin{proof}
The two claims in \eqref{E:3.2.3} follow immediately from \eqref{E:3.2.2a} and the assumption that $B\ge0$.
So all we need to show in order to complete the proof is that 
\begin{equation}\label{E:3.2.3.1}
B=\frac{\partial A_2}{\partial x_1}-\frac{\partial A_1}{\partial x_2}\,.
\end{equation}
But plugging in the definition \eqref{E:3.2.2a} of $A$, we get that
\begin{equation}
\begin{aligned}
\frac{\partial A_2}{\partial x_1}(\vec x)-\frac{\partial A_1}{\partial x_2}(\vec x)
&=2\int_0^1\! tB(t\vec x)\,dt\\
&\quad+\int_0^1\! t^2\bigl(x_1\partial_1B(t\vec x)+x_2\partial_2B(t\vec x)\bigr)\,dt\\
&=\int_0^1\! \frac{\partial}{\partial t}\bigl(t^2B(t\vec x)\bigr)\,dt=B(\vec x)\,,
\end{aligned}
\end{equation}
as claimed.
\end{proof}

Given this, we notice the importance of the quantity
\begin{equation}\label{E:3.2.4}
a(\vec x)=\int_0^1\! tB(t\vec x)\,dt:
\end{equation}

\begin{lemma}\label{L:3.0.2}
For smooth functions $a$ and $B$ on $|\vec x|<1$, we have that
\begin{equation}\label{E:3.2.5}
a(\vec x)=\int_0^1\!\! tB(t\vec x)\,dt\quad\Longleftrightarrow\quad B(\vec x)=2a(\vec x)+\vec x\cdot\nabla a(\vec x)\,.
\end{equation}
Furthermore, $a$ is radially symmetric iff $B$ is, and in this case the equivalence \eqref{E:3.2.5} becomes
\begin{equation}\label{E:3.2.6}
a(r)=\int_0^1\!\! tB(tr)\,dt \quad\Longleftrightarrow\quad B(r)=2a(r)+ra'(r)\,.
\end{equation}
\end{lemma}

\begin{proof}
The proof of the equivalence from \eqref{E:3.2.5} consists of two applications of the calculation from the proof of 
Lemma~\ref{L:3.0.1}. Now, if $B$ is radially symmetric, and we set $r=|\vec x|$, then we see directly from \eqref{E:3.2.5} that $a$ is also, and
$$
a(r)=\int_0^1\!\! tB(tr)\,dt\,.
$$
Conversely, if $a$ is radially symmetric, notice that
\begin{equation*}
\vec x\cdot \nabla a(\vec x)=ra'(r)\,,
\end{equation*}
and so it follows again from \eqref{E:3.2.5} that $B$ is radially symmetric also, and
$$
B(r)=2a(r)+ra'(r)\,.
$$
\end{proof}

Finally, it is important to rewrite the Hamiltonian in polar coordinates:
\begin{lemma}\label{L:3.1}
Let $B$ be a (smooth) magnetic field on $\ID$, $A$ the vector potential in the transversal gauge (i.e., 
chosen as in Lemma~\ref{L:3.0.1}), and $a$ as in \eqref{E:3.2.5}. Then, if $(r,\theta)$ are the polar coordinates associated to the 
rectangular coordinates $\vec x=(x_1,x_2)$, we have that
\begin{equation}\label{E:3.2.7}
\bigl(-i\nabla-A\bigr)^2=-\frac1r\,\frac{\partial}{\partial r}\,r\,\frac{\partial}{\partial r}+
\left(-\frac{i}{r}\,\frac{\partial}{\partial\theta}-ra(r,\theta)\right)^2\,.
\end{equation}
\end{lemma}

\begin{proof}
The identity \eqref{E:3.2.7} follows by a direct change of variables. First recall that
\begin{equation*}
A(\vec x)=a(\vec x)\,\bigl(-x_2\quad x_1\bigr)\,,
\end{equation*}
so by expanding the square we see that
\begin{equation*}
\bigl(-i\nabla-A\bigr)^2=-\Delta+ia\,(-x_2\quad x_1)\cdot\nabla+i(-x_2\quad x_1)\cdot\nabla\,a+|\vec x|^2 a^2\,.
\end{equation*} 
The standard expression in polar coordinates of the Laplacian in 2 dimensions is
\begin{equation*}\label{E:3.2.8}
\Delta=\frac{\partial^2}{\partial r^2}+\frac1r\,\frac{\partial}{\partial r}+\frac{1}{r^2}\,\frac{\partial^2}{\partial\theta^2}\,,
\end{equation*} 
and we notice that
\begin{equation*}\label{E:3.2.9}
\begin{aligned}
(-x_2\quad x_1)\cdot\nabla
&=-r\sin\theta\left(\cos\theta\,\frac{\partial}{\partial r}-\frac{\sin\theta}{r}\,\frac{\partial}{\partial\theta}\right)+
r\cos\theta\left(\sin\theta\,\frac{\partial}{\partial r}+\frac{\cos\theta}{r}\,\frac{\partial}{\partial\theta}\right)\\
&=\sin^2\theta\,\frac{\partial}{\partial\theta}+\cos^2\theta\,\frac{\partial}{\partial\theta}=\frac{\partial}{\partial\theta}\,.
\end{aligned}
\end{equation*}
Plugging these two expressions into the expanded square above leads directly to~\eqref{E:3.2.7}.
\end{proof}


\end{document}